\documentclass[11pt]{article}
\usepackage[margin=1in]{geometry}
\usepackage[utf8]{inputenc}
\usepackage[T1]{fontenc}
\usepackage{times}

\usepackage[dvipsnames]{xcolor} 
\definecolor{MyGreen}{rgb}{0, 0.7, 0}
\definecolor{MyRed}{rgb}{0.8, 0, 0}
\usepackage{soul}
\usepackage{url}
\usepackage{hyperref}
\hypersetup{
colorlinks=true,
citecolor=teal, 
urlcolor=blue, 
linkcolor=MyRed, 
}
\usepackage{amsmath,amsthm,amssymb,amsfonts,mathtools,bbm}
\DeclarePairedDelimiter\ceiling{\lceil}{\rceil} 
\DeclarePairedDelimiter\floor{\lfloor}{\rfloor} 
\usepackage{booktabs}
\usepackage{helvet}
\usepackage{courier}
\usepackage{graphicx}
\usepackage{nicefrac}

\usepackage[ruled,linesnumbered,vlined]{algorithm2e}
\usepackage{comment,xspace,enumitem}
\usepackage{booktabs}
\usepackage{wrapfig}
\usepackage{cleveref} 
\usepackage[normalem]{ulem}
\theoremstyle{plain}
\newtheorem{theorem}{Theorem}[section]
\newtheorem{lemma}[theorem]{Lemma}
\newtheorem{proposition}[theorem]{Proposition}
\newtheorem{example}[theorem]{Example}
\newtheorem{claim}[theorem]{Claim}

\theoremstyle{definition}
\newtheorem{definition}[theorem]{Definition}

\theoremstyle{remark}
\newtheorem{remark}[theorem]{\upshape\bfseries Remark}

\usepackage[square]{natbib}
\usepackage{tikz}
\usetikzlibrary{positioning,shapes,arrows,graphs,quotes,automata,backgrounds}

\newcommand{\E}{\mathbb{E}}
\newcommand{\indicator}[1]{\mathbbm{1}_{\left\{#1\right\}}\xspace}
\newcommand{\committeeSize}{k}
\newcommand{\approvalBallotOfAgent}[1]{A_{#1}} 
\newcommand{\randomizedCommittee}{\mathbf{X}}

\newcommand{\GCR}{\text{GCR}\xspace}
\newcommand{\MES}{\text{MES}\xspace}

\newcommand{\EJR}{\text{EJR}\xspace}
\newcommand{\EJRplus}{\text{EJR+}\xspace}

\newcommand{\RD}{\text{MES}\xspace}
\newcommand{\budgetMES}{b}
\newcommand{\BWMES}{BW-MES\xspace}
\newcommand{\BWMESFull}{Best-of-Both-Worlds MES\xspace}
\newcommand{\BWGCR}{BW-GCR\xspace}
\newcommand{\AllocFromShares}{\textsc{AllocationFromShares}\xspace}

\title{Best-of-Both-Worlds Fairness in Committee Voting}
\author{
Haris Aziz \and
Xinhang Lu \and
Mashbat Suzuki \and
Jeremy Vollen \and
Toby Walsh \and \\
UNSW Sydney \\
\{haris.aziz, xinhang.lu, mashbat.suzuki, j.vollen, t.walsh\}@unsw.edu.au}
\date{}

\begin{document}
\maketitle

\begin{abstract}
The best-of-both-worlds paradigm advocates an approach that achieves desirable properties both ex-ante and ex-post.
We launch a best-of-both-worlds fairness perspective for the important social choice setting of approval-based committee voting.
To this end, we initiate work on ex-ante proportional representation properties in this domain and formalize a hierarchy of notions including Individual Fair Share (IFS), Unanimous Fair Share (UFS), Group Fair Share (GFS), and their stronger variants.
We establish their compatibility with well-studied ex-post concepts such as extended justified representation (EJR) and fully justified representation (FJR).
Our first main result is a polynomial-time algorithm that simultaneously satisfies ex-post EJR, ex-ante GFS and ex-ante Strong UFS.
Subsequently, we strengthen our ex-post guarantee to FJR and present an algorithm that outputs a lottery which is ex-post FJR and ex-ante Strong UFS, but does not run in polynomial time.

\end{abstract}

\section{Introduction}
\label{sec:intro}

Fairness is one of the central concerns when aggregating the preferences of multiple agents.
Just as envy-freeness is viewed as a central fairness goal when allocating resources among agents~\citep{Fole67,Moul19a}, proportional representation is a key fairness desideratum when making collective choice such as selecting a set of alternatives~\citep{LaSk23a}.
However, in both contexts, fairness is often too hard to achieve perfectly as an outcome satisfying their respective fairness notions may not exist.

Two successful approaches to counter the challenge of non-existence of fair outcomes are \emph{relaxation} and \emph{randomization}.
The idea of relaxation is to weaken the ideal notion of fairness enough to get meaningful and guaranteed existence of fair outcomes.
In the resource allocation context, a widely pursued relaxation of envy-freeness is envy-freeness up to one item (EF1)~\citep{LMMS04,Budish11}.
In the social choice context of approval-based committee voting, \emph{core} is viewed as the strongest proportional representation concept.
Since it is not known whether a core stable outcome is guaranteed to exist, researchers have focused on natural relaxations of the core that are guaranteed to exist and are based on the idea of \emph{justified representation}~\citep{ABC+17}.

The second approach to achieve fairness is via randomization that specifies a probability distribution (or \emph{lottery}) over ex-post integral outcomes.
Randomization is one of the oldest tools to achieve fairness and has been applied to contexts such as resource allocation~\citep{BoMo01a} and collective choice~\citep{BMS05a}.

In most of the work on resource allocation and social choice, the two major approaches of relaxation and randomization are pursued separately.
The concerted focus on pursuing both approaches simultaneously to achieve good fairness guarantees is a recent phenomenon~\citep{Aziz19a}, and has been referred to as the \emph{best-of-both-worlds} paradigm~\citep{FSV20b}.

While the best-of-both-worlds approach has been applied to fair resource allocation, it has not been explored as much in the contexts of social choice and public decision-making.
Furthermore, this perspective has never been taken with respect to approval-based committee voting, which is the setting considered in this paper.

In approval-based committee voting, each voter approves of a subset of candidates.
Based on these expressed approvals, a committee (i.e., a subset of candidates) of a target size is selected.
Almost all of the papers on fairness in approval-based committee voting focus on ex-post fairness guarantees such as \emph{justified representation (JR)}, \emph{proportional justified representation (PJR)}~\citep{SEL+17}, \emph{extended justified representation (EJR)}~\citep{ABC+17}, and \emph{fully justified representation (FJR)~\citep{PPS21a}}.
For instance, EJR says that for each positive integer~$\ell$, if a group of at least $\ell \cdot n / \committeeSize$ voters approve at least~$\ell$ common candidates, some voter in the group must have at least~$\ell$ approved candidates in the committee.
Here, $n$ is the number of voters and $\committeeSize$ is the target committee size.
To the best of our knowledge, the only work that uses randomization to obtain ex-ante fairness in approval-based committee voting is that of \citet{CJMW20}.
This work, however, ignores ex-post fairness guarantees.

We initiate a best-of-both-worlds fairness perspective in the context of social choice, particularly approval-based committee voting.
We first motivate our approach by noting that without randomization, one cannot guarantee that each voter get strictly positive expected representation, and thus fails a property known as \emph{positive share}.\footnote{\label{ft:integral-not-give-positive-share}
This can be seen from a simple example with a target committee size of one and two voters with disjoint approvals.}
In this paper, we will define a hierarchy of ex-ante fairness properties stronger than positive share and give a class of randomized algorithms which achieve these properties in addition to the existing ex-post fairness properties.

\subsection{Our Contributions}

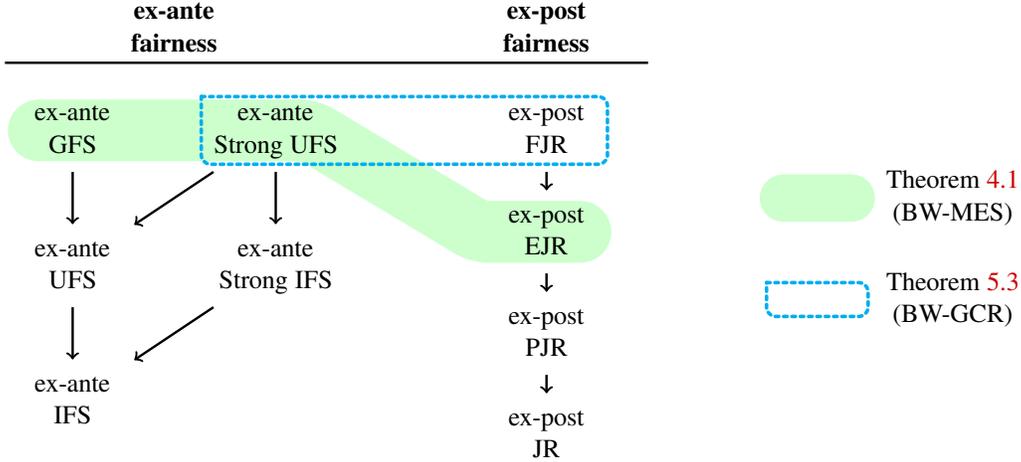
\begin{figure}[t]
\centering
\scalebox{0.9}{
\begin{tikzpicture}
\tikzstyle{onlytext}=[]
\tikzset{venn circle/.style={circle,minimum width=0mm,fill=#1,opacity=0.1}}

\node[onlytext] (ex-ante fair) at (1.5,0) {\begin{tabular}{c}\textbf{ex-ante}\\\textbf{fairness}\end{tabular}};
\draw[-, line width=1pt] (-1,-0.5) -- (8.5,-0.5) ;

\node[onlytext] (ex-ante GFS) at (0,-1.5) {\begin{tabular}{c}{ex-ante}\\{GFS}\end{tabular}};
\node[onlytext] (ex-ante UFS) at (0,-3.5) {\begin{tabular}{c}{ex-ante}\\{UFS}\end{tabular}};
\node[onlytext] (ex-ante IFS) at (0,-5.5) {\begin{tabular}{c}{ex-ante}\\{IFS}\end{tabular}};

\node[onlytext] (ex-ante sUFS) at (3,-1.5) {\begin{tabular}{c}{ex-ante}\\{Strong UFS}\end{tabular}};
\node[onlytext] (ex-ante sIFS) at (3,-3.5) {\begin{tabular}{c}{ex-ante}\\{Strong IFS}\end{tabular}};

\node[onlytext] (ex-post fair) at (7,0) {\begin{tabular}{c}\textbf{ex-post}\\\textbf{fairness}\end{tabular}};

\node[onlytext] (ex-post FJR) at (7,-1.5) {\begin{tabular}{c}{ex-post}\\{FJR}\end{tabular}};
\node[onlytext] (ex-post EJR) at (7,-3) {\begin{tabular}{c}{ex-post}\\{EJR}\end{tabular}};
\node[onlytext] (ex-post PJR) at (7,-4.5) {\begin{tabular}{c}{ex-post}\\{PJR}\end{tabular}};
\node[onlytext] (ex-post JR) at (7,-6) {\begin{tabular}{c}{ex-post}\\{JR}\end{tabular}};

\draw[->, line width=1pt] (ex-ante GFS) -- (ex-ante UFS);
\draw[->, line width=1pt] (ex-ante UFS) -- (ex-ante IFS);

\draw[->, line width=1pt] (ex-ante sUFS) -- (ex-ante sIFS);
\draw[->, line width=1pt] (ex-ante sUFS) -- (ex-ante UFS);
\draw[->, line width=1pt] (ex-ante sIFS) -- (ex-ante IFS);

\draw[->, line width=1pt] (ex-post FJR) -- (ex-post EJR);
\draw[->, line width=1pt] (ex-post EJR) -- (ex-post PJR);
\draw[->, line width=1pt] (ex-post PJR) -- (ex-post JR);

\begin{pgfonlayer}{background}
\draw[line width=26pt,green!20,line cap=round,rounded corners] (-0.5,-1.5) -- (3.5,-1.5) -- (6,-3) -- (7.5,-3);
\draw[-,cyan, dotted, line width=1.5pt, rounded corners, line cap=round] (1.9,-1) -- (7.9,-1) -- (7.9,-2) -- (1.9,-2) -- (1.9,-1);

\draw[line width=20pt,green!20,line cap=round,rounded corners] (10.5,-2.5) -- (11.5,-2.5);
\node[onlytext] () at (13,-2.5) {\begin{tabular}{c}\Cref{thm:GFS+EJR}\\(\BWMES)\end{tabular}};

\draw[-,cyan, dotted, line width=1.5pt, rounded corners, line cap=round] (10.25,-3.75) -- (11.75,-3.75) -- (11.75,-4.25) -- (10.25,-4.25) -- (10.25,-3.75);
\node[onlytext] () at (13,-4) {\begin{tabular}{c}\Cref{thm:sUFS+FJR}\\(\BWGCR)\end{tabular}};
\end{pgfonlayer}
\end{tikzpicture}
}
\caption{Visualization of ex-ante and ex-post fairness hierarchies studied in this paper.
An arrow from (A) to (B) denotes that (A) implies (B).
We highlight those properties which can be satisfied simultaneously and point to the corresponding theorems in the key.}
\label{fig:relations}
\end{figure}

Our first contribution is to broaden the best-of-both-worlds fairness paradigm, which has so far been limited to resource allocation, and explore it in the context of social choice problems, specifically committee voting.
Whereas the relaxed notions of ex-post fairness have been examined at great length for committee voting, the literature on ex-ante fairness is less developed.

In \Cref{sec:fairness-fractional-committees}, we formalize several natural axioms for ex-ante fairness that are careful extensions of similar concepts proposed in the restricted setting of single-winner voting.
These include the following concepts in increasing order of strength as well as stronger variants: \emph{positive share},
\emph{individual fair share (IFS)}, \emph{unanimous fair share (UFS)}, and \emph{group fair share (GFS)}.
For instance, positive share simply requires that each voter expects to have non-zero probability of having some approved candidate selected.
At the other end of the spectrum, GFS gives a desirable level of ex-ante representation to \emph{every} coalition of voters.
In \Cref{fig:relations}, we present ex-ante and ex-post fairness hierarchies and establish logical relations between them.

In line with the goals of the best-of-both-worlds fairness paradigm, our central research question is to understand which combinations of ex-ante and ex-post fairness properties can be achieved simultaneously.
In \Cref{sec:GFS+EJR}, we show that ex-ante GFS, ex-ante Strong UFS and ex-post EJR can be achieved simultaneously by devising a class of randomized algorithms.
Our class of algorithms, which we call \emph{\BWMESFull} (\emph{\BWMES}), uses as a subroutine the well-known Method of Equal Shares (MES)~\citep{PeSk20a}.
Furthermore, a lottery satisfying GFS, Strong UFS and EJR can be computed in polynomial time via an algorithm belonging to \BWMES.

Finally, we examine whether it is possible to achieve a stronger ex-post guarantee while maintaining some ex-ante fairness properties.
In \Cref{sec:BWGCR}, we answer this question affirmatively by presenting an algorithm (referred to as \BWGCR) which outputs a lottery that is ex-post FJR and ex-ante Strong UFS.
Our algorithm combines both the Greedy Cohesive Rule~\citep{PPS21a} and the BW-MES rule.

\subsection{Related Work}

In this paper, we examine approval-based committee voting, a generalization of the classical voting setting which has been studied at length, particularly from the 19th century to the present.
One of the persistent questions within the committee voting setting is how to produce committees which proportionally represent groups of voters.
\citet{ABC+17} initiated an axiomatic study of approval-based committee voting based on the idea of ``justified representation'' for cohesive groups.
The study has led to a hierarchy of axioms and a large body of work focusing on voting rules which produce committees satisfying these axioms and thus give some guarantee of fair representation \citep{AEH+18a,EFI+22,BFJL23}.
For a detailed survey of the recent work on approval-based committee voting, we refer the readers to the book of \citet{LaSk23a}.
While this paper also targets committees satisfying these properties, we examine outcomes that are randomized committees which specify a probability distribution over integral committees.

In social choice theory, randomization is one of the oldest tools used to achieve stronger fairness properties and to bypass various impossibility results which apply to discrete outcomes~\citep[see, e.g.,][]{Gibb77a}.
For single-winner randomized voting (also known as \emph{probabilistic voting}) with approval preferences, \citet{BMS05a} defined ex-ante fairness notions, the individual fair share (IFS) and unanimous fair share (UFS), and provide rules satisfying them.
They also proposed a group fairness property called group fair share (GFS)~\citep[see][]{BMS02a}, independently proposed by \citet{Dudd15a}, which is stronger than UFS and IFS but weaker than core fair share (CFS), a group fairness and stability property inspired by that of \emph{core} from cooperative game theory \citep{Scar67a}.
In a setting which generalizes probabilistic voting by allowing arbitrary endowments, \citet{BBPS21} studied fair distribution rules and introduce GFS, which they show is equivalent to a notion called decomposability.
None of their results imply those presented in this paper since their setting places no restriction on the distribution a single alternative can receive our setting does.

\citet{MPS20a} explored the trade-off between group fairness and utilitarian social welfare by measuring the ``price of fairness'' with respect to fairness axioms such as IFS and GFS.
We formulate these ex-ante fairness properties for the more general committee voting setting for the first time.
As mentioned, we search for outcomes which also give fair representation to groups ex-post, a desideratum which has no analogue in the classical voting setting.

\citet{Aziz19a} proposed research directions regarding probabilistic decision making with desirable ex-ante and ex-post stability or fairness properties.
\citet{FSV20b} were the first to coin the term ``best-of-both-worlds fairness''.
They examined the compatibility of achieving ex-ante envy-freeness and ex-post near envy-freeness in the context of resource allocation.
There have been several recent works on best-of-both-worlds fairness in resource allocation~\citep{HPPS20a,BEF21a,BEF22,AFS+23a,AGM23a,HSV23a,FMNP23}.
Other works that consider the problem of implementing a fractional allocation over deterministic allocations subject to constraints include~\citep{BCKM12a,AkNi20a}.

\section{Preliminaries}
\label{sec:prelim}

For any positive integer~$t \in \mathbb{N}$, let $[t] \coloneqq \{1, 2, \dots, t\}$.
Let $C = [m]$ be the set of \emph{candidates} (also called \emph{alternatives}).
Let $N = [n]$ be the set of voters.
We assume that the voters have \emph{approval} (also known as \emph{dichotomous} or \emph{binary}) preferences, that is, each voter~$i \in N$ approves a non-empty ballot~$\approvalBallotOfAgent{i} \subseteq C$.
We denote by~$N_c$ the set of voters who approve of candidate~$c$, i.e., $N_c\coloneqq \{i \in N \mid c \in \approvalBallotOfAgent{i}\}$.
An \emph{instance}~$I$ can be described by a set of candidates~$C$, a list of ballots $\mathcal{A} = (\approvalBallotOfAgent{1}, \approvalBallotOfAgent{2}, \dots, \approvalBallotOfAgent{n})$, and a positive committee size~$\committeeSize \leq m$ which is an integer.

\paragraph{Integral and Fractional Committees}
As is standard in committee voting, an \emph{(integral) winning committee}~$W$ is a subset of~$C$ having size~$\committeeSize$.
A \emph{fractional committee} is specified by an $m$-dimensional vector~$\vec{p} = (p_c)_{c \in C}$ with $p_c \in [0, 1]$ for each~$c \in C$, and $\sum_{c \in C} p_c = \committeeSize$.
Note an integral committee~$W$ can be alternatively represented by the vector~$\vec{p}$ in which $p_c = 1$ for all~$c \in W$ and $p_c = 0$ otherwise.
For notational convenience, let $\vec{1}_W \in \{0, 1\}^m$, whose $j^{\text{th}}$ component is~$1$ if and only if~$j \in W$, be the vector representation of an integral committee~$W$.
The utility of voter~$i \in N$ for a (fractional or integral) committee~$\vec{p}$ is given by $u_i(\vec{p}) \coloneqq \sum_{c \in \approvalBallotOfAgent{i}} p_c$.

\paragraph{Randomized Committees}
A \emph{randomized committee}~$\randomizedCommittee$ is a lottery over integral committees and specified by a set of~$s \in \mathbb{N}$ tuples $\{(\lambda_j, W_j)\}_{j \in [s]}$ with $\sum_j \lambda_j = 1$, where for each~$j \in [s]$, the integral committee $W_j \subseteq C$ is selected with probability $\lambda_j \in [0, 1]$.
The \emph{support} of~$\randomizedCommittee$ is the set of integral committees $\{W_1, W_2, \dots, W_s\}$.
Unless specified otherwise, when we simply say ``a committee'', it will mean an integral committee.

A randomized committee $\{(\lambda_j, W_j)\}_{j \in [s]}$ is an \emph{implementation} of (or ``implements'') a fractional committee~$\vec{p}$ if $\vec{p} = \sum_{j \in [s]} \lambda_j \vec{1}_{W_j}$.
Note that there may exist many implementations of any given fractional committee.
Given a randomized committee $\randomizedCommittee = (\lambda_j, W_j)_{j \in [r]}$ implementing a fractional outcome~$\vec{p}$, we can interpret the fractional utility as the expected utility of the randomized committee, i.e., $u_i(\vec{p}) = \E_{W \sim \randomizedCommittee}[u_i(W)]$.

The fact that any fractional committee can be implemented by a probability distribution over integral committees of the same size is implied by various works on randomized rounding schemes in combinatorial optimization~\citep{Grim04a,GKPS06,ALM+19a}.
We explain this connection explicitly using the classical result of \citet{GKPS06} and frame it in our context.
Theorem~2.3 of \citet{GKPS06} states that there is a polynomial-time rounding scheme to \emph{sample} an integral committee from a randomized committee satisfying three properties.
The first property ensures that the randomized committee is a valid implementation of the fractional committee.
The second property ensures that each committee in the support of the implementation are of size~$\committeeSize$.
We do not need the third property for our purposes.
While one can \emph{sample} an outcome from a support in polynomial time, the support could have exponential size.
Since we want an explicit construction of the randomized committee, such a rounding scheme does not run in polynomial time due to the exponential-size output.
In order to have polynomial-time computation, we resort to randomized rounding schemes which output an explicit distribution over a support of polynomial size, for example, the \AllocFromShares of \citet{ALM+19a}.

\subsection{Fairness for Integral Committees}

Fairness properties for integral committees are well-studied in committee voting.
A desideratum that has received significant attention is \emph{justified representation (JR)}~\citep{ABC+17}.
In order to reason about JR and its strengthenings, an important concept is that of a cohesive group.
For any positive integer~$\ell$, a group of voters $N' \subseteq N$ is said to be \emph{$\ell$-cohesive} if $|N'| \geq \ell \cdot n/k$ and $\left| \bigcap_{i \in N'} \approvalBallotOfAgent{i} \right| \geq \ell$.

\begin{definition}[JR]
\label{def:JR}
A committee~$W$ is said to satisfy \emph{justified representation (JR)} if for every $1$-cohesive group of voters~$N' \subseteq N$, it holds that $\approvalBallotOfAgent{i} \cap W \neq \emptyset$ for some~$i \in N'$.
\end{definition}

Two important strengthenings of JR have been proposed.

\begin{definition}[PJR~\citep{SEL+17}]
\label{def:PJR}
A committee~$W$ is said to satisfy \emph{proportional justified representation (PJR)} if for every positive integer~$\ell$ and every $\ell$-cohesive group of voters~$N' \subseteq N$, it holds that $\left| \left( \bigcup_{i \in N'} \approvalBallotOfAgent{i} \right) \cap W \right| \geq \ell$.
\end{definition}

\begin{definition}[EJR~\citep{ABC+17}]
\label{def:EJR}
A committee~$W$ is said to satisfy \emph{extended justified representation (EJR)} if for every positive integer~$\ell$ and every $\ell$-cohesive group of voters~$N' \subseteq N$, it holds that $|\approvalBallotOfAgent{i} \cap W| \geq \ell$ for some~$i \in N'$.
\end{definition}

It follows directly from the definitions that EJR implies PJR, which in turn implies JR.
A committee providing EJR (and therefore PJR and JR) always exists and can be computed in polynomial time~\citep{ABC+17,PeSk20a}.

\section{Fairness for Fractional Committees}
\label{sec:fairness-fractional-committees}

In this section, we first introduce fairness properties for fractional committees, after which we establish the relations between these fractional fairness notions and the integral fairness notions presented in \Cref{sec:prelim}.

\subsection{Fairness Concepts}

Whereas the literature on fairness concepts for integral committees is very well-developed, fairness properties for fractional committees are largely unexplored except for the special case of single-winner voting~\citep{BMS05a,Dudd15a,ABM20a}.
We introduce a hierarchy of fairness notions for fractional committees in the committee voting setting by generalizing axioms from the single-winner context based on \emph{fair share}.
The weakest in the hierarchy of axioms is \emph{individual fair share (IFS)}, the idea behind which is that ``each one of the $n$ voters `owns' a $1/n$-th share of decision power, so she can ensure an outcome she likes with probability at least~$1/n$'', as \citet[page~18:2]{ABM20a} put it.
This idea suggests at least two distinct interpretations of the utility lower bound guaranteed by IFS:
\begin{enumerate}[label=(\alph*)]
\item\label{enum:control-integral-outcome}
each voter is given $1/n$ probability to choose their favourite integral outcome, or

\item\label{enum:control-fractional-outcome}
each voter can select $1/n$ of the (fractional) outcome.
\end{enumerate}
In probabilistic voting, as we will see shortly, both interpretations coincide.
Critically, this is not the case in the committee voting setting. Instead, these interpretations diverge and lead to two alternative hierarchies of fair share axioms for committee voting, which we term \emph{fair share} and \emph{strong fair share}, respectively.

We begin by defining both generalizations of individual fair share (IFS). Both impose a natural lower bound on individual utilities stronger than that of \emph{positive share}, which requires that $u_i(\vec{p}) > 0$.
In the single-winner setting, IFS requires that the probability that the (single) alternative selected is approved by any individual voter is no less than $1/n$.
It is thus tempting to require $u_i(\vec{p}) = \sum_{c \in \approvalBallotOfAgent{i}} p_c \geq \frac{k}{n}$, which turns out to be too strong in our setting as a fractional committee satisfying it may not exist.\footnote{For instance, let~$\committeeSize > n$ and consider the case where voter $i$ only approves a single candidate.
Then, the above inequality cannot hold for $i$ as the left-hand side is upper bounded by $|A_i| = 1$ while the right-hand side is greater than one and can be arbitrarily large.}
Intuitively speaking, this is because our only restriction on the voters' approval sets is that each voter approves of at least one candidate, just as is standard in the single-winner literature.
However, whereas in the $\committeeSize = 1$ special case this assumption is sufficient to ensure that a uniform cut-off utility lower bound for each voter is feasible, the same is not true for general~$\committeeSize$.

\begin{definition}[IFS]
\label{def:IFS}
A fractional committee~$\vec{p}$ satisfies \emph{IFS} if for each~$i \in N$,
\[
u_i(\vec{p}) = \sum_{c \in \approvalBallotOfAgent{i}} p_c \geq \frac{1}{n} \cdot \min \{\committeeSize, |\approvalBallotOfAgent{i}|\}.
\]
\end{definition}

While IFS captures interpretation~\ref{enum:control-integral-outcome} of fair share, Strong IFS reflects interpretation~\ref{enum:control-fractional-outcome} which says that each voter should control $1/n$ of the fractional outcome.

\begin{definition}[Strong IFS]
\label{def:sIFS}
A fractional committee~$\vec{p}$ satisfies \emph{Strong IFS} if for each~$i \in N$,
\[
u_i(\vec{p}) = \sum_{c \in \approvalBallotOfAgent{i}} p_c \geq \min \left\{ \frac{\committeeSize}{n}, |\approvalBallotOfAgent{i}| \right\}.
\]
\end{definition}

Next, we strengthen IFS (resp., Strong IFS) to \emph{unanimous fair share (UFS)} (resp., \emph{Strong UFS}), which guarantees any group of like-minded voters an influence proportional to its size.

\begin{definition}[UFS]
\label{def:UFS}
A fractional committee~$\vec{p}$ is said to provide \emph{UFS} if for any~$S \subseteq N$ where $\approvalBallotOfAgent{i} = \approvalBallotOfAgent{j}$ for any~$i, j \in S$, then the following holds for each $i \in S$:
\[
u_i(\vec{p}) = \sum_{c \in \approvalBallotOfAgent{i}} p_c \geq \frac{|S|}{n} \cdot \min \{\committeeSize, |\approvalBallotOfAgent{i}|\}.
\]
\end{definition}

\begin{definition}[Strong UFS]
\label{def:sUFS}
A fractional committee~$\vec{p}$ is said to provide \emph{Strong UFS} if for any~$S \subseteq N$ where $\approvalBallotOfAgent{i} = \approvalBallotOfAgent{j}$ for any~$i, j \in S$, then the following holds for each $i\in S$:
\[
u_i(\vec{p}) = \sum_{c \in \approvalBallotOfAgent{i}} p_c \geq \min \left\{ |S| \cdot \frac{\committeeSize}{n}, |\approvalBallotOfAgent{i}| \right\}.
\]
\end{definition}

Our primary focus in this paper is a stronger notion---\emph{group fair share (GFS)}---which gives a non-trivial ex-ante representation guarantee to \emph{every} coalition of voters.

\begin{definition}[GFS]
\label{def:GFS}
A fractional committee~$\vec{p}$ is said to provide \emph{GFS} if the following holds for every~$S \subseteq N$:
\[
\sum_{c \in \bigcup_{i \in S} \approvalBallotOfAgent{i}} p_c \geq \frac{1}{n} \cdot \sum_{i \in S} \min\{\committeeSize, |\approvalBallotOfAgent{i}|\}.
\]
\end{definition}

We note that a GFS fractional committee always exists and can be computed by a very natural algorithm called \emph{Random Dictator}, which selects each voter's favourite integral committee (breaking ties arbitrarily) with probability~$1/n$.

\begin{proposition}
Random Dictator computes a randomized committee that is ex-ante GFS in polynomial time.
\end{proposition}

\begin{proof}
First, it is clear that Random Dictator runs in polynomial time.
Let~$\{(\frac{1}{n}, W_i)\}_{i \in N}$ be the randomized committee returned by Random Dictator for an instance of our problem.
Let~$\vec{p}$ be the fractional committee it implements.
Note that $p_c = \sum_{i \in N} \frac{1}{n} \cdot \indicator{c \in W_i}$ for all~$c \in C$, where $\indicator{\cdot}$ is an indicator function.

Fix any~$S \subseteq N$.
Substituting to the LHS of the GFS guarantee, we get
\begin{align*}
\sum_{c \in \bigcup_{i \in S} A_i} p_c &= \sum_{c \in \bigcup_{i \in S} A_i} \left( \sum_{j \in N} \frac{1}{n} \cdot \indicator{c \in W_j} \right) \\
&\geq \frac{1}{n} \cdot \sum_{j \in S} \sum_{c \in \bigcup_{i \in S} A_i} \indicator{c \in W_j} \\
&= \frac{1}{n} \cdot \sum_{j \in S} \Big| W_j \cap \bigcup_{i \in S} A_i \Big| \geq \frac{1}{n} \cdot \sum_{j \in S} \min \{k, |A_i|\},
\end{align*}
where the last transition holds because $W_j$ is one of voter $j$'s most preferred committees by the definition of Random Dictator.
\end{proof}

However, Random Dictator does not satisfy Strong IFS.\footnote{To see this, consider an instance with $k = 2$, three candidates $\{c_1, c_2, c_3\}$, and two voters with $A_1 = \{c_1\}$ and $A_2 = \{c_2, c_3\}$.
Since each voter must select an integral committee, voter~$1$ allocates some of her probability to a candidate she does not approve, and thus $\sum_{c \in A_1} p_c = p_{c_1} = 1/2 < \min \left\{ \frac{\committeeSize}{n}, |\approvalBallotOfAgent{1}| \right\} = 1$.}
Indeed, this is the principal reason we chose to name the respective axiom hierarchies as we did.
There is a significant precedent of treating the Random Dictator mechanism as the utility lower bound for fair share axioms, including by the authors \citep[page~167]{BMS05a} who introduced fair share:
\begin{quote}
\itshape
Fair welfare share uses the random dictator mechanism as the disagreement option that each participant is entitled to enforce.
\end{quote}

Furthermore, the natural extensions of Strong UFS to Strong GFS are not guaranteed to exist.
For instance, following~\citep{BMS05a,BBPS21} and our own \Cref{def:GFS}, we may be tempted to formulate the RHS of Strong GFS as the sum of the Strong IFS guarantees, i.e., $\sum_{i\in S} \min \left\{\frac{\committeeSize}{n}, |\approvalBallotOfAgent{i}|\right\}$.
However, \Cref{ex:bad_gfs} will show that a fractional committee satisfying this notion may not always exist.\footnote{This demonstrates that the setting of \citet{BBPS21} deviates from our setting as \Cref{alg:GFS+EJR} outputs a decomposable outcome satisfying Strong IFS, but Strong GFS is not satisfied.}
Another natural generalization would be the following:
\begin{equation}
\label{eq:bad_gfs}
\sum_{c \in \bigcup_{i \in S} \approvalBallotOfAgent{i}} p_c \geq \min \left\{ |S| \cdot \frac{\committeeSize}{n}, \left| \bigcup_{i\in S} A_i \right| \right\}
\end{equation}

\Cref{eq:bad_gfs} captures the spirit of strong fair share well by affording each coalition of voters control over the outcome proportional to their size, upper bounded by the number of candidates they collectively approve.
\Cref{ex:bad_gfs} below shows the formulation of Strong GFS given by \Cref{eq:bad_gfs} is also impossible to satisfy.

\begin{example}
\label{ex:bad_gfs}
Consider an instance with $n = 4$, $k = 4$, and the following approval sets:
\[
A_1 = A_2 = \{c_1\} \qquad A_3 = \{c_1, c_2, c_3\} \qquad A_4 = \{c_1, c_4, c_5\}.
\]
For the group $T = \{1, 2, 3\}$, \Cref{eq:bad_gfs} requires that
\[
\sum_{c \in \bigcup_{i \in T} A_i} p_c \geq \min \left\{ |T| \cdot \frac{\committeeSize}{n}, \left| \bigcup_{i \in T} A_i \right| \right\} = 3.
\]
This means that each candidate in~$A_3$ must receive probability~$1$.
By symmetry, the same for the group $\{1, 2, 4\}$ and thus $A_4$.
However, since $|A_3 \cup A_4| = 5$ and $k = 4$, this is an impossibility.
\end{example}

It follows directly from the definitions that GFS implies UFS, which in turn implies IFS, and that each of our generalizations of IFS, UFS, and GFS correspond to their definitions in the single-winner voting scenarios.
The relations between these axioms are pictured in \Cref{fig:relations}.
We would also like to remark that any pair of ex-ante fairness notions not given explicit logical relation in \Cref{fig:relations} are logically independent (i.e., neither implies the other).
The details are deferred to \Cref{app:relations-ex-ante-fairness}.

\subsection{Relations between Fractional and Integral Fairness Concepts}

Before describing and proving our approach to best-of-both-worlds fairness in this setting, we first investigate the logical relations between our ex-ante and ex-post properties for integral committees.
In doing so, we rule out some naive approaches to our problem of interest and illustrate the usefulness of our ex-ante properties.
We begin by remarking that, as mentioned in \Cref{sec:intro}, there may not exist an integral committee satisfying any of our ex-ante fairness properties.

\begin{remark}
\label{rmk:integral-not-give-positive-share}
An integral committee satisfying positive share may not exist.
\end{remark}

As mentioned, this fact is the principal motivation for studying randomized committees.
However, we would also like to understand what our fairness concepts for fractional committees can tell us about the space of integral committees satisfying our ex-post fairness properties.
The following example and summarizing remark show that our fractional fairness concepts can help in reasoning about which outcomes satisfying our ex-post properties are more desirable.

\begin{example}
\label{ex:EJRnotIFS}
Consider an instance with $n = 10$ and $\committeeSize = 4$.
Suppose eight of the voters approve of candidates $\{c_1, c_2, c_3, c_4\}$ and the remaining two voters approve of candidate $\{c_5\}$.

Note that the committee $W = \{c_1, c_2, c_3, c_4\}$ satisfies EJR.
This is because $4 \cdot \frac{10}{4} > 8 \geq 3 \cdot \frac{10}{4}$ and~$W$ already includes at least three candidates approved by the eight voters.
Also, since $2 < 1 \cdot \frac{10}{4}$, \EJR does not guarantee the two voters who approve~$\{c_5\}$ being represented in~$W$, violating positive share.
The alternative committee of $\{c_1, c_2, c_3, c_5\}$ also satisfies EJR, and additionally satisfies IFS.
\end{example}

\begin{remark}
As shown by \Cref{ex:EJRnotIFS}, even when an integral committee satisfying IFS exists, some EJR outcomes may not satisfy positive share.
\end{remark}

From this, we conclude that a successful algorithm must select carefully from the space of outcomes satisfying our ex-post properties. We next explore to what extent our ex-ante properties imply our ex-post properties in the integral case.

\begin{proposition}
\label{prop:ifs_jr}
If an integral committee satisfies IFS, then it satisfies JR.
\end{proposition}

\begin{proof}
Let $W$ be an integral committee which satisfies IFS and let $\vec{p}=\vec{1}_W$. Then, for all $i\in N$, we have
$$u_i(\vec{p}) = \sum_{c\in A_i} p_c = \vert A_i \cap W \vert \geq \frac{\min ( \vert A_i \vert, k )}{n} > 0.$$
Thus, since $|A_i\cap W|$ is an integer, $|A_i\cap W| \geq 1$ for all $i\in N$ and it follows that $W$ is JR.
\end{proof}

While \Cref{prop:ifs_jr} hints at a synergy between our ex-ante and ex-post properties, \Cref{prop:gfs_notpjr} below shows that even the strongest ex-ante property in our hierarchy does not imply the next strongest ex-post property.

\begin{proposition}
\label{prop:gfs_notpjr}
If an integral committee satisfies GFS, it does not necessarily satisfy PJR.
\end{proposition}

\begin{proof}
Consider an instance with $k = 4$ and $n = 2$ and the following approval profile:
\[
A_1 = \{c_1, c_2\} \qquad A_2 = \{c_2, c_3, c_4, c_5\}.
\]
The committee $W = \{c_2, c_3, c_4, c_5\}$ satisfies GFS since $|W \cap A_1| = 1 = \frac{1}{n} \cdot \min \{k, |A_1|\}$ (and voter~$2$ receives their most preferred committee).
Now see that $\{1\}$ is a 2-cohesive group; however, $|A_1 \cap W| < 2$, meaning that $W$ does not satisfy PJR.
\end{proof}

Despite this negative finding, in the following section, we will present a class of algorithms computing randomized committees which simultaneously satisfy ex-ante GFS and ex-post EJR.

\section{Best of Both Worlds: GFS + Strong UFS + EJR}
\label{sec:GFS+EJR}

In this section, we present a family of rules called \emph{\BWMESFull} (or \emph{\BWMES} for short), which obtains best-of-both-worlds fairness.
Our main result is:

\begin{theorem}
\label{thm:GFS+EJR}
\BWMES (\Cref{alg:GFS+EJR}) outputs a randomized committee that is ex-ante GFS, ex-ante Strong UFS, and ex-post EJR.
Furthermore, the algorithm can be implemented in polynomial time.
\end{theorem}

\begin{algorithm}[t]
\caption{\BWMES: Ex-ante GFS, Ex-ante Strong UFS, and Ex-post EJR}
\label{alg:GFS+EJR}
\DontPrintSemicolon

\KwIn{Voters~$N = [n]$, candidates~$C = [m]$, approvals~$(\approvalBallotOfAgent{i})_{i \in N}$, and committee size~$\committeeSize$.}
\KwOut{A GFS and Strong UFS fractional committee~$\vec{p} = (p_c)_{c \in C}$ and its implementation as a lottery over EJR integral committees.}

\BlankLine
Initialize $\budgetMES_i \gets \committeeSize / n$ for each~$i \in N$, i.e., the budget voter~$i$ can spend on buying candidates.\;
Initialize $y_{ij} \gets 0$ for each~$i \in N$ and~$j \in C$, i.e., the amount voter~$i$ spends on candidate~$j$.\; \label{ALG:GFS+EJR:initialize-payments}

\BlankLine
\tcp{Obtain an integral EJR committee via MES.}
Let $W_\MES$ be an integral EJR committee returned by the first phase of MES~\citep[see, e.g.,][Rule~11]{LaSk23a} with initial budget $(\budgetMES_i)_{i \in N}$.\; \label{ALG:GFS+EJR:MES}
Update~$(\budgetMES_i)_{i \in N}$ to be the \emph{remaining} budgets of the voters after executing MES.\; \label{ALG:GFS+EJR:remaining-budget}
Update~$y_{ij}$ for each~$i \in N$ and~$j \in W_\MES$ to be the amount each voter~$i$ spent on candidate~$j$ during MES.\;
$\vec{p} = (p_1, p_2, \dots, p_m) \gets \vec{1}_{W_\MES}$ \tcp*{Initialize a fractional committee.} \label{ALG:GFS+EJR:initialize-fractional-committee}

\BlankLine
\tcp{Extend the integral EJR committee to a fractional committee that is GFS and Strong UFS.}
$N' \gets \{i \in N \mid \approvalBallotOfAgent{i} \setminus W_\MES \neq \emptyset\}$\;
\ForEach{$i \in N'$}{
	Voter~$i$ spends an arbitrary amount of~$y_{ic}$ on each~$c \in \approvalBallotOfAgent{i} \setminus W_\MES$ such that $\sum_{c \in \approvalBallotOfAgent{i} \setminus W_\MES} y_{ic} = \budgetMES_i$. \label{ALG:GFS+EJR:fund-approved-candidates}
}
\ForEach{$i \in N \setminus N'$}{
	Voter~$i$ spends~$\budgetMES_i$ in any fashion provided $p_c \leq 1$ for any~$c \in C$. Update~$y_{ic}$ accordingly. \label{ALG:GFS+EJR:fund-arbitrary-candidates}
}

\BlankLine
\tcp{Implementation.}
Apply a randomized rounding scheme~\citep[e.g.,][]{ALM+19a} to~$\vec{p}$, which outputs a lottery over integral committees of size~$k$. Let $\{(\lambda_j, W_j)\}_{j \in [s]}$ denote the randomized committee.\; \label{ALG:GFS+EJR:rounding}

\Return{$\vec{p}$ and its implementation~$\{(\lambda_j, W_j)\}_{j \in [s]}$}
\end{algorithm}

\subsection{A Family of Rules: \BWMES}

We start by providing an intuition behind our family of rules \BWMES, whose pseudocode can be found in \Cref{alg:GFS+EJR}.
At a high level, \BWMES follows in spirit the idea of the Method of Equal Shares (MES) of \citet{PeSk20a}.
To be more precise, we follow the MES algorithm description of \citet[Rule~11]{LaSk23a}, and make use of its first phase.
For ease of exposition, we simply refer to this first phase as ``MES''.

\begin{definition}[\MES]
Each voter is initially given a budget of~$\committeeSize/n$ to spend on candidates, each of which costs~$1$.
We start with an empty outcome~$W = \emptyset$ and sequentially add candidates to~$W$.
To add a candidate~$j$ to~$W$, we need the voters who approve~$j$ to pay for it.
Write~$y_{i j}$ for the amount that voter~$i$ pays for~$j$ and $b_i = \committeeSize/n - \sum_{j \in W} y_{i j} \geq 0$ for the amount of budget voter~$i$ has left.
For $\rho \geq 0$, we say that a candidate~$c \in C \setminus W$ is \emph{$\rho$-affordable} if
\[
\sum_{i \in N_c} \min(b_i, \rho) = 1.
\]
If no candidate is $\rho$-affordable for any~$\rho$, \MES terminates and return~$W$.
Otherwise, it adds to~$W$ a candidate~$j \in C \setminus W$ that is $\rho$-affordable for a minimum~$\rho$; payments are given by $y_{ij} = \min(b_i, \rho)$.
\end{definition}

\Cref{alg:GFS+EJR} returns an \emph{integral} EJR committee~$W_\MES$ in \cref{ALG:GFS+EJR:MES}.
Denote by $(b_i)_{i \in N}$ the \emph{remaining} budget of the voters after executing MES (\cref{ALG:GFS+EJR:remaining-budget}).
Our next step is to extend~$W_\MES$ to a \emph{fractional} GFS committee of size~$\committeeSize$ using voters' remaining budget.
We first initialize a fractional committee~$\vec{p}$ using~$W_\MES$ in \cref{ALG:GFS+EJR:initialize-fractional-committee}.
It is worth noting that for any~$c \in C \setminus W_\MES$, $\sum_{i \in N_c} \budgetMES_i < 1$; otherwise candidate~$c$ would have been included in~$W_\MES$ in \cref{ALG:GFS+EJR:MES}.
The key idea behind our completion method for the fractional committee in this family of rules is the following:
\begin{itemize}
\item We first let each~$i \in N$ such that $\approvalBallotOfAgent{i} \setminus W_\MES \neq \emptyset$ spend her remaining budget~$\budgetMES_i$ on candidates~$\approvalBallotOfAgent{i} \setminus W_\MES$, in an arbitrary way.

\item Next, for any other voter, her remaining budget can be spent on any candidate~$c \in C$ provided~$p_c \leq 1$.
\end{itemize}

Finally, for the implementation step, we can use any rounding method that implements the fractional committee~$\vec{p}$ by randomizing over integral committees of the same size; see, e.g., the \AllocFromShares method of \citet{ALM+19a}, the stochastic method of \citet{Grim04a}, or the dependent randomized rounding scheme of \citet{GKPS06}.

In the following, we use an illustrative example to demonstrate our algorithm.

\begin{example}\label{ex:BW-MES}
The following committee voting instance is used in Example~2.12 of \citet{LaSk23a} to illustrate MES.
Let~$\committeeSize = 3$.
Consider the following approval preferences which involve four candidates:
\begin{alignat*}{2}
&\approvalBallotOfAgent{1} = \approvalBallotOfAgent{2} = \approvalBallotOfAgent{3} = \{c_3, c_4\} \qquad &&\approvalBallotOfAgent{4} = \approvalBallotOfAgent{5} = \{c_1, c_2\} \\
&\approvalBallotOfAgent{6} = \approvalBallotOfAgent{7} = \{c_1, c_3\} \qquad &&\approvalBallotOfAgent{8} = \{c_2, c_4\}.
\end{alignat*}

The voters start with a budget of~$3/8$.
\Cref{ALG:GFS+EJR:MES} of \Cref{alg:GFS+EJR} returns candidates~$\{c_1, c_3\}$ (alternatively, $\vec{p} = (1, 0, 1, 0)$).
We give more details below.
In the first round, since every voter has the same amount of initial budget and candidate~$c_3$ has the most approving voters, \MES selects~$c_3$ and has each approving voter~$\{1, 2, 3, 6, 7\}$ pay~$1/5$.
Agents' budgets are updated as follows:
\begin{alignat*}{2}
&\budgetMES_1 = \budgetMES_2 = \budgetMES_3 = 7/40 \qquad && \budgetMES_4 = \budgetMES_5 = 3/8 \\
&\budgetMES_6 = \budgetMES_7 = 7/40 \qquad &&\budgetMES_8 = 3/8.
\end{alignat*}
In the next round, we calculate the $\rho$-affordability for each of the remaining candidates:
\begin{itemize}
\item $c_1$ is $\frac{13}{40}$-affordable because $\sum_{i \in \{4, 5, 6, 7\}} \min(b_i, 13/40) = 1$;
\item $c_2$ is $\frac{1}{3}$-affordable because $\sum_{i \in \{4, 5, 8\}} \min(b_i, 1/3) = 1$;
\item $c_4$ is not $\rho$-affordable for any~$\rho$ due to $\sum_{i \in \{1, 2, 3, 8\}} b_i = 36/40 < 1$.
\end{itemize}
\MES selects candidate~$c_1$, and then terminates as no remaining candidate is affordable.

The remaining budgets of the voters after \MES returns~$\vec{p} = (1, 0, 1, 0)$ are as follows:
\begin{alignat*}{2}
&\budgetMES_1 = \budgetMES_2 = \budgetMES_3 = 7/40 \qquad && \budgetMES_4 = \budgetMES_5 = 1/20 \\
&\budgetMES_6 = \budgetMES_7 = 0 \qquad &&\budgetMES_8 = 3/8.
\end{alignat*}
Then, in \cref{ALG:GFS+EJR:fund-approved-candidates} of \Cref{alg:GFS+EJR}, each voter~$i \in [8]$ spends~$\budgetMES_i$ on candidates~$\{c_2, c_4\} \cap \approvalBallotOfAgent{i}$ lexicographically.
Therefore, we obtain the fractional committee $\vec{p} = (1, 19/40, 1, 21/40)$, which can be implemented by the randomized committee $\left\{ \left( \frac{19}{40}, \{c_1, c_2, c_3\} \right), \left( \frac{21}{40}, \{c_1, c_3, c_4\} \right) \right\}$, e.g., using the \AllocFromShares method of \citet{ALM+19a}.

In the following, we explain in more detail how \AllocFromShares outputs the lottery.
\AllocFromShares takes as input the vector $\vec{p} = (1, 19/40, 1, 21/40)$, and first reorders the candidates as $c_1, c_3, c_2, c_4$ (in ascending order of $p_j - \floor{p_j}$), yielding $\vec{s} = (1, 1, 19/40, 21/40)$.
Although our target committee size is~$3$, the selection probability we need to allocate via rounding is~$\alpha = 1$, computed in line~5 of \AllocFromShares (by summing $s_j - \floor{s_j}$ over~$j$).
To be consistent with our own notation as well, we initialize $(\lambda_1, \lambda_2, \lambda_3, \lambda_4) \gets \vec{0}$, where~$\lambda_i$ is the probability associated with the resulting rounding allocation when considering~$s_i$.

We start with $\texttt{low} = 1$ and $\texttt{high} = 4$.
The total probability~$\bar{p}$ allocated so far is~$0$.
We begin by considering the following allocation:
\[
\vec{t}_1 = (\ceiling{1}, \floor{1}, \floor{19/40}, \floor{21/40}) = (1, 1, 0, 0).
\]
Since $0 = s_1 - \floor{s_1} - \lambda_1 < \ceiling{s_4} - s_4 - \bar{p} + \lambda_4 = 19/40$, this allocation will get a probability of~$\lambda_1 = 0$.
Now, $\texttt{low} = 2$ and everything else remains the same.
We look at the next allocation:
\[
\vec{t}_2 = (\floor{1}, \ceiling{1}, \floor{19/40}, \floor{21/40}) = (1, 1, 0, 0).
\]
Since $0 = s_2 - \floor{s_2} - \lambda_2 < \ceiling{s_4} - s_4 - \bar{p} + \lambda_4 = 19/40$, this allocation will get a probability of~$\lambda_2 = 0$.
Now, $\texttt{low} = 3$ and everything else remains the same.
We move to the following allocation:
\[
\vec{t}_3 = (\floor{1}, \floor{1}, \ceiling{19/40}, \floor{21/40}) = (1, 1, 1, 0).
\]
Since $19/40 = s_3 - \floor{s_3} - \lambda_3 \geq \ceiling{s_4} - s_4 - \bar{p} + \lambda_4 = 19/40$, this allocation (i.e., $\{c_1, c_3, c_2\}$) will get a probability of~$\lambda_3 = 19/40$.
Now, $\texttt{high} = 3$, $\alpha = 0$, and $\bar{p} = 19/40$.
Finally, we look at:
\[
\vec{t}_4 = (\floor{1}, \floor{1}, \floor{19/40}, \ceiling{21/40}) = (1, 1, 0, 1).
\]
Since $\alpha = 0$, this allocation (i.e., $\{c_1, c_3, c_4\}$) will get a probability of~$\lambda_4 = 21/40$.
\end{example}

\subsection{Analysis of \BWMES}

Before proving \Cref{thm:GFS+EJR}, we first show a lower bound on voters' utilities provided by \BWMES.

\begin{claim}
\label{claim:LB-spending}
In \Cref{alg:GFS+EJR}, for each~$i \in N$, it holds that
\[
\sum_{j \in \approvalBallotOfAgent{i}} y_{ij} \geq \frac{1}{n} \cdot \min \{\committeeSize, |\approvalBallotOfAgent{i}|\}.
\]
\end{claim}

\begin{proof}
Recall that each~$i \in N$ is given an initial budget of~$\committeeSize / n$.
Fix any~$i \in N$ such that $\approvalBallotOfAgent{i} \setminus W_\MES \neq \emptyset$.
By the construction of \Cref{alg:GFS+EJR}, voter~$i$ spends her budget~$\committeeSize / n$ on candidates that she approves, in \cref{ALG:GFS+EJR:MES} when executing MES or in \cref{ALG:GFS+EJR:fund-approved-candidates} when completing the fractional committee~$\vec{p}$.
It thus follows that $\sum_{j \in \approvalBallotOfAgent{i}} y_{ij} = \frac{\committeeSize}{n} \geq \frac{1}{n} \cdot \min \{\committeeSize, |\approvalBallotOfAgent{i}|\}$.

Now, fix any~$i \in N$ such that $\approvalBallotOfAgent{i} \setminus W_\MES = \emptyset$ (alternatively, $\approvalBallotOfAgent{i} \subseteq W_\MES$).
In other words, all candidates approved by voter~$i$ are already \emph{fully} included in the fractional committee~$\vec{p}$.
Clearly, $|\approvalBallotOfAgent{i}| \leq \committeeSize$.
Fix any~$c \in \approvalBallotOfAgent{i}$.
Recall that $N_c$ consists of voters who approve candidate~$c$.
If voter~$i$ spends the remainder of their budget on candidate~$c$ for any $c \in A_i$, then the claim holds trivially.
Otherwise, by the construction of MES, voter~$i$ pays an amount of at least~$1 / |N_c|$ for candidate~$c$, meaning that $y_{ic} \geq 1 / |N_c| \geq 1/n$.
We thus have $\sum_{j \in \approvalBallotOfAgent{i}} y_{ij} \geq \frac{1}{n} \cdot |\approvalBallotOfAgent{i}| = \frac{1}{n} \cdot \min \{\committeeSize, |\approvalBallotOfAgent{i}|\}$, as desired.
\end{proof}

We are now ready to establish our main result.

\begin{proof}[Proof of \Cref{thm:GFS+EJR}]
We break the proof into the following five parts.

\paragraph{Feasibility}
For each~$c \in C$, note that whenever a (positive) amount~$p_c$ of the candidate is added to the fractional committee~$\vec{p}$, the voters together pay a total of~$p_c$.
Since the voters have a total starting budget of~$\committeeSize$ and each spends their entire budget, \Cref{alg:GFS+EJR} returns a fractional committee of size~$\committeeSize$.
Next, due to the randomized rounding scheme we use in \cref{ALG:GFS+EJR:rounding}, each integral committee in the returned randomized committee is of size~$\committeeSize$.
In short, the fractional committee~$\vec{p}$ and each integral committee in the randomized committee returned by \Cref{alg:GFS+EJR} respect the size constraint.

\paragraph{Ex-ante GFS}
Recall that~$y_{ij}$ denotes the amount each voter~$i \in N$ spent on each candidate~$j \in C$ in \Cref{alg:GFS+EJR}.
The fractional committee~$\vec{p} = (p_1, p_2, \dots, p_m)$ can thus be alternatively expressed by $p_j = \sum_{i \in N} y_{ij}$ for each~$j \in C$.

Given this, for any~$S \subseteq N$, we now have
\[
\sum_{j \in \bigcup_{v \in S} \approvalBallotOfAgent{v}} p_j
= \sum_{j \in \bigcup_{v \in S} \approvalBallotOfAgent{v}} \sum_{i \in N} y_{ij}
\geq \sum_{i \in S} \sum_{j \in \bigcup_{v \in S} \approvalBallotOfAgent{v}} y_{ij}
\geq \sum_{i \in S} \sum_{j \in \approvalBallotOfAgent{i}} y_{ij}
\geq \sum_{i \in S} \frac{1}{n} \cdot \min \{\committeeSize, |\approvalBallotOfAgent{i}|\},
\]
where the last transition is due to \Cref{claim:LB-spending}.

\paragraph{Ex-ante Strong UFS}
Consider any $S \subseteq N$ such that $A_i = A$ for all~$i \in S$, where $A \subseteq C$.
First, if $A \subseteq W_\MES$, Strong UFS follows immediately.
In the following, we consider the case where $A \setminus W_\MES \neq \emptyset$.
In this case, each voter in~$S$ spends the entirety of their budget in \cref{ALG:GFS+EJR:fund-approved-candidates}.
Thus, since voters only pay for candidates they approve, for each $i \in S$, we have
\[
u_i(\vec{p}) \geq \sum_{j \in A} \sum_{i \in S} y_{ij} = |S| \cdot \frac{k}{n}.
\]

\paragraph{Ex-post EJR}
According to the first property of \citet[Theorem~2.3]{GKPS06}, $W_\MES$ is included in every realization.
Since \MES satisfies EJR~\citep{PeSk20a}, we have that the lottery output is ex-post EJR.

\paragraph{Polynomial-time Computation}
To begin, note that both MES of \citet{PeSk20a} and the randomized rounding scheme of \citet{GKPS06} used in \cref{ALG:GFS+EJR:MES,ALG:GFS+EJR:rounding} run in polynomial time.
Thus, the computational complexity of any rule in the \BWMES family is dominated by how it completes the fractional committee in \crefrange{ALG:GFS+EJR:fund-approved-candidates}{ALG:GFS+EJR:fund-arbitrary-candidates}, which can be done in polynomial time as follows.
Iterate through candidates~$C \setminus W_\MES$ in an arbitrary order, allocating to each candidate the remaining budgets of those agents who approve the candidate, i.e., for $c \in C\setminus W_\MES$, $p_c \gets \sum_{i \in N_c} b_i$, and zero the agents' budgets accordingly.
\end{proof}

Recently, \citet{BrPe23} proposed a new notion called \emph{\EJRplus}, which strengthens \EJR, has guaranteed existence, and moreover, can be verified in polynomial time.

\begin{definition}[\EJRplus]
A committee~$W$ is said to satisfy~\emph{\EJRplus} if there is no candidate~$c \in C \setminus W$, group of voters~$N' \subseteq N$, and~$\ell \in \mathbb{N}$ with $|N'| \geq \ell \cdot \frac{n}{\committeeSize}$ such that
\[
c \in \bigcap_{i \in N'} \approvalBallotOfAgent{i} \quad \text{and} \quad |\approvalBallotOfAgent{i} \cap W| < \ell \text{ for all } i \in N'.
\]
\end{definition}

\citet{BrPe23} showed that \MES satisfies \EJRplus.
Together with a similar argument we used in the proof of \Cref{thm:GFS+EJR}, it can be seen that the randomized committee returned by \Cref{alg:GFS+EJR} also satisfies the stronger ex-post \EJRplus property.
We thus can strengthen \Cref{thm:GFS+EJR} as follows:

\begin{proposition}
\BWMES (\Cref{alg:GFS+EJR}) outputs a randomized committee that is ex-ante GFS, ex-ante Strong UFS, and ex-post \EJRplus.
Furthermore, the algorithm can be implemented in polynomial time.
\end{proposition}

\subsection{Completing MES}

Because MES may return an EJR committee~$W_\MES$ of size \emph{less than}~$\committeeSize$, several ways of extending $W_\MES$ to an \emph{integral} committee of size exactly~$\committeeSize$ have been discussed~\citep[see, e.g.,][]{PeSk20a,LaSk23a}.
As we have seen previously (\Cref{rmk:integral-not-give-positive-share}), an EJR committee may not provide positive share, let alone GFS.
We provide a novel perspective on the completion of MES.
Specifically, we define a family of rules which extends an integral committee returned by MES to a \emph{randomized} committee providing GFS.

Due to the flexibility in the definition of the \BWMES family, the number of \BWMES rules obtaining distinct outcomes can be quite large.
For example, one \BWMES rule that seems quite natural is that which continues in the spirit of MES: for the candidate whose supporters have the most collective budget leftover, this budget is spent on the candidate, and we continue in this fashion sequentially.
It is an interesting future direction to further identify specific algorithms in the \BWMES family which provide additional desiderata such as high social welfare or additional ex-ante properties.

\section{Best of Both Worlds: Strong UFS + FJR}
\label{sec:BWGCR}

In this section, we present an algorithm which satisfies an alternative strengthening of EJR, known as \emph{fully justified representation (FJR)}, although our result comes at the cost of ex-ante GFS and polynomial time computation.

We begin by defining an ex-post fairness guarantee stronger than \EJR, first introduced by \citet{PPS21a} in the context of \emph{participatory budgeting}.
For our purpose, we state the axiom and \citeauthor{PPS21a}'s results pertaining to FJR in the context of approval-based committee voting.

\begin{definition}[FJR~\citep{PPS21a,LaSk23a}]
Given a positive integer $\beta$ and a set of candidates~$T \subseteq C$, a group of voters~$N' \subseteq N$ is said to be \emph{weakly $(\beta, T)$-cohesive} if $|N'| \geq |T| \cdot \frac{n}{\committeeSize}$ and $|A_i \cap T| \geq \beta$ for every voter~$i \in N'$.

A committee~$W$ is said to satisfy \emph{fully justified representation (FJR)} if for every weakly $(\beta, T)$-cohesive group of voters~$N'$, it holds that $|A_i\cap W| \geq \beta$ for some $i\in N'$.
\end{definition}

Note that FJR is incomparable to \EJRplus~\citep{BrPe23}.
\citet{PPS21a} gave an algorithm called \emph{Greedy Cohesive Rule (GCR)} which computes an FJR committee.

\begin{definition}[\GCR]
The Greedy Cohesive Rule (GCR) begins by marking all voters as \emph{active} and initializing $W = \emptyset$.
In each step, GCR searches for a set of voters~$N' \subseteq N$ who are all active and a set of candidates~$T \subseteq C \setminus W$ such that $N'$ is weakly $(\beta, T)$-cohesive, choosing the sets which maximize $\beta$, breaking ties in favour of smaller~$|T|$, followed by larger~$|N'|$.
GCR then adds the candidates from~$T$ to~$W$ and labels all of the voters in~$N'$ as inactive.
If, at any step, no weakly $(\beta, T)$-cohesive group exists for any positive integer~$\beta$, then the algorithm returns~$W$ and terminates.
\end{definition}

\subsection{The Algorithm: \BWGCR}

\begin{algorithm}[t]
\caption{\BWGCR: Ex-ante Strong UFS and ex-post FJR}
\label{alg:sUFS+FJR}
\DontPrintSemicolon

\KwIn{Voters~$N = [n]$, candidates~$C = [m]$, approvals~$(\approvalBallotOfAgent{i})_{i \in N}$, and committee size~$\committeeSize$.}
\KwOut{A Strong UFS fractional committee~$\vec{p} = (p_c)_{c \in C}$ and its implementation as a lottery over FJR integral committees.}

\BlankLine
$W_\GCR \gets \operatorname{GreedyCohesiveRule}(N, C, (\approvalBallotOfAgent{i})_{i \in N}, k)$\; \label{alg:sUFS+FJR:gcr}
Denote by~$r$ the number of steps GCR executes.\;
\ForEach{$j \in [r]$}{
	Denote by~$N_j$ the weakly cohesive group selected in the $j$-th step of GCR and~$\beta_j$ the corresponding~$\beta$ value used in this step.\;
	Denote by $\{N_j^1, \ldots, N_j^{\eta_j}\}$ the~$\eta_j$ many distinct unanimous groups within~$N_j$.
}
Denote by~$N_\GCR$ the set of inactive voters.\;
$N_\RD \gets N \setminus N_\GCR$\;
$k_\RD \gets \committeeSize - |W_\GCR|$\;

\BlankLine
\ForEach{$j \gets 1$ \KwTo $r$}{
	\ForEach{$z \gets 1$ \KwTo $\eta_j$}{
		$b_i \gets \max \left\{ 0, \frac{k}{n} - \frac{\beta_j}{\vert N_j^z \vert} \right\}$ for all $i \in N_j^z$ \label{alg:sUFS+FJR:b-jz}
	}
}

\ForEach{$i \in N_\RD$}{
	$b_i \gets \frac{1}{|N_\RD|} \cdot \left( k_\RD - \sum_{l \in N_\GCR} b_l \right)$ \label{alg:sUFS+FJR:b-rd}
}

\BlankLine
Execute \Cref{alg:GFS+EJR} from \cref{ALG:GFS+EJR:initialize-payments} with budgets $(b_i)_{i\in N}$, candidates $C\setminus W_\GCR$, and committee size $k_\RD$ to get a randomized outcome $\{(\lambda_j, W_j)\}_{j \in [s]}$. \;\label{alg:sUFS+FJR:call-bwmes}
\Return{$(\boldsymbol{\lambda}, \mathbf{W}) = \{(\lambda_j, W_j\cup W_\GCR)\}_{j \in [s]}$ and the fractional outcome it implements, $\vec{p}$.}
\end{algorithm}

We give an algorithm, referred to as \BWGCR, which obtains ex-ante Strong UFS in addition to ex-post FJR.
The pseudocode of the algorithm is presented as \Cref{alg:sUFS+FJR}.
At a high level, our algorithm is an interesting synthesis of the Greedy Cohesive Rule and our \BWMES Rule (\Cref{alg:GFS+EJR}).
A key ingredient here is how we set voters' budgets for the \BWMES phase in order to obtain a feasible Strong UFS fractional committee in the end.

More specifically, our algorithm begins by calling \GCR, denoting by~$r$ the number of steps it executes before terminating.
For each~$j \in [r]$, we refer to~$N_j$, $T_j$, and~$\beta_j$ as the values of~$N'$, $T$, and~$\beta$ for the weakly cohesive group selected in the $j$-th step of \GCR.
Let~$W_\GCR = \bigcup_{j \in [r]} T_j$ be the returned integral FJR committee, $N_\GCR = \bigcup_{j \in [r]} N_j$ the set of voters who are inactivated (put differently, represented), and $N_\RD = N \setminus N_\GCR$ the set of active voters.
For ease of exposition, we sometimes discuss a specific unanimous group within a given weakly cohesive group.
Given any set of voters~$S$, we can partition the voters into a collection of \emph{(maximal) unanimous groups}~$\{S^1, S^2, \dots, S^\eta\}$ such that for each~$z \in [\eta]$, voters~$S^z$ have an identical preference (i.e., they are unanimous) and for any~$i \in S^z$ and $i' \in S^{z'}$ with $z \neq z'$, $\approvalBallotOfAgent{i} \neq \approvalBallotOfAgent{i'}$ (i.e., each unanimous group is maximal).
We will denote by~$N_j^z$ the $z$-th unanimous group of the $j$-th weakly cohesive group~$N_j$ encountered by \Cref{alg:sUFS+FJR} and will refer to the number of distinct unanimous groups in~$N_j$ as~$\eta_j$.

Next, \Cref{alg:sUFS+FJR} proceeds by carefully setting an individual budget~$b_i$ for each voter~$i \in N_\GCR$ in \cref{alg:sUFS+FJR:b-jz}, and next for voters in~$N_\RD$ in \cref{alg:sUFS+FJR:b-rd} in a different way.
We then call \BWMES (\Cref{alg:GFS+EJR}) as a subroutine with budgets~$(b_i)_{i \in N}$, candidates~$C \setminus W_\GCR$ and committee size~$\committeeSize_\RD = \committeeSize - |W_\GCR|$, and get a randomized outcome~$\{(\lambda_j, W_j)\}_{j \in [s]}$.
We will show shortly that the randomized committee~$\{(\lambda_j, W_j \cup W_\GCR)\}_{j \in [s]}$ is feasible, ex-post FJR, and ex-ante Strong UFS.

\subsection{Analysis of \BWGCR}

We now state the main result of this section, which shows that \Cref{alg:sUFS+FJR} provides best-of-both-worlds fairness.

\begin{theorem} \label{thm:sUFS+FJR}
\BWGCR (\Cref{alg:sUFS+FJR}) computes a randomized committee that is ex-ante Strong UFS and ex-post FJR.
\end{theorem}

We begin by proving some auxiliary statements related to the budgets~$b_i$ for each~$i \in N$. These statements also provide a rationale for the particular way in which budgets are set in \cref{alg:sUFS+FJR:b-jz} and \cref{alg:sUFS+FJR:b-rd}.
The first two claims concern voters in~$N_\GCR$, for whom the budget is set in \cref{alg:sUFS+FJR:b-jz}.
Our first claim suggests that if a unanimous group is given zero budget, this unanimous group has already received representation amounting to its fair share.

\begin{claim}
\label{claim:b-jz:zero}
$\forall j \in [r], z \in [\eta_j], i \in N_j^z$, $b_i = 0 \iff \beta_j\geq |N_j^z| \cdot \frac{k}{n}$.
\end{claim}

\begin{proof}
This follows since $b_i = 0 \iff \frac{k}{n} - \frac{\beta_j}{\vert N_j^z \vert} \leq 0 \iff \frac{\vert N_j^z \vert k}{n} \leq \beta_j$.
\end{proof}

Our next claim establishes a connection between~$\beta_j$ and~$|T_j|$ for any weakly cohesive group such that some voter is assigned positive budget.

\begin{claim}
\label{claim:b-jz:positive}
For any~$j \in [r]$, if there exists~$i \in N_j$ such that $b_i > 0$, then $\beta_j = |T_j|$.
\end{claim}

\begin{proof}
Fix any $j \in [r], z \in [\eta_j], i \in N_j^z$ such that $b_i > 0$.
We have
\[
b_i > 0 \implies \frac{k}{n} - \frac{\beta_j}{|N_j^z|} > 0 \implies |N_j^z| > \beta_j \cdot \frac{n}{k}.
\]

Let $A_j^z$ denote the approval set of voters~$N_j^z$.
Since $N_j$ is weakly $(\beta_j, T_j)$-cohesive, by definition, $|A_j^z \cap T_j| \geq \beta_j$.
Now, let $T_j^z \subseteq A_j^z \cap T_j$ be of size exactly~$\beta_j$, i.e., $|T_j^z| = \beta_j$.
It can be verified that $N_j^z$ is weakly $(\beta_j, T_j^z)$-cohesive.
Thus, it must be that $T_j^z = T_j$, since otherwise $|T_j^z| < |T_j|$ and \GCR would have selected~$N_j^z$, $\beta_j$ and~$T_j^z$ in step~$j$ according to the tie-breaking rule.
We conclude that $\beta_j = |T_j|$.
\end{proof}

While it is apparent from \cref{alg:sUFS+FJR:b-jz} that each voter~$i \in N_\GCR$ gets non-negative budget~$b_i$, it is unclear at first glance whether voters~$N_\RD$ also receive non-negative budgets.
Below, we show a (stronger) lower bound on the budgets for~$N_\RD$.

\begin{lemma}
\label{lem:brd-lower-bound}
For all~$i \in N_\RD$, $b_i \geq \frac{\committeeSize}{n}$.
\end{lemma}

\begin{proof}
Since each~$i \in N_\RD$ is given the same amount of budget in \cref{alg:sUFS+FJR:b-rd}, it suffices to show that $k_\RD - \sum_{l \in N_\GCR} b_l \geq |N_\RD| \cdot \frac{\committeeSize}{n}$.

For ease of exposition, in this proof, we \emph{re-order} the~$r$ weakly cohesive groups encountered by \GCR in \cref{alg:sUFS+FJR:gcr}.
Let~$t \in \{0, 1, \dots, r\}$ be an index such that all voters in~$\bigcup_{j = t+1}^r N_j$ receive zero budget in \cref{alg:sUFS+FJR:b-jz}, i.e., for all~$i \in \bigcup_{j = t+1}^{r} N_j$, $b_i = 0$.
Moreover, for each~$N_j$, we let the first~$\eta'_j \in \{0, 1, 2, \dots, \eta_j\}$ unanimous groups be the ones which receive positive budget in \cref{alg:sUFS+FJR:b-jz}.
Note that $\eta_j^\prime \geq 1$ for $j\in[t]$ due to how we re-order the weakly cohesive groups.

We now proceed to show the inequality, where the second transition is due to how we re-order the weakly cohesive groups as well as how we label the unanimous groups within each group:
\allowdisplaybreaks
\begin{align*}
\committeeSize_\RD - \sum_{l \in N_\GCR} b_l
&= \committeeSize_\RD - \sum_{l \in N_\GCR} \max \left\{ 0, \frac{\committeeSize}{n} - \frac{\beta_j}{|N_j^z|} \right\} \\
&= \committeeSize_\RD - \sum_{j = 1}^t \sum_{z = 1}^{\eta'_j} |N_j^z| \cdot \left( \frac{\committeeSize}{n} - \frac{\beta_j}{|N_j^z|} \right) \\
&= \committeeSize_\RD - \frac{\committeeSize}{n} \cdot \sum_{j = 1}^t \sum_{z = 1}^{\eta'_j} |N_j^z| + \sum_{j = 1}^t \sum_{z = 1}^{\eta'_j} \beta_j \\
&= \committeeSize_\RD - \frac{\committeeSize}{n} \cdot \sum_{j = 1}^t \sum_{z = 1}^{\eta'_j} |N_j^z| + \sum_{j = 1}^t \sum_{z = 1}^{\eta'_j} |T_j| \tag{$\because$ \Cref{claim:b-jz:positive}} \\
&\geq \committeeSize_\RD - \frac{\committeeSize}{n} \cdot \sum_{j = 1}^t \sum_{z = 1}^{\eta'_j} |N_j^z| + \sum_{j = 1}^t |T_j| \tag{$\because \eta'_j \geq 1$} \\
&= \committeeSize_\RD - \frac{\committeeSize}{n} \cdot \left( |N_\GCR| - \sum_{j = 1}^t \sum_{z = \eta'_j + 1}^{\eta_j} |N_j^z| - \sum_{j = t + 1}^r |N_j| \right) + \sum_{j = 1}^t |T_j| \\
&\geq \committeeSize_\RD - \frac{\committeeSize}{n} \cdot |N_\GCR| + \frac{\committeeSize}{n} \cdot \sum_{j = t + 1}^r |N_j| + \sum_{j = 1}^t |T_j| \\
&\geq \committeeSize_\RD - \frac{\committeeSize}{n} \cdot |N_\GCR| + \sum_{j = t+1}^r |T_j| + \sum_{j = 1}^t |T_j| \tag{$\because |N_j| \geq |T_j| \cdot \frac{n}{\committeeSize}$} \\
&= \committeeSize - |W_\GCR| - \frac{\committeeSize}{n} \cdot |N_\GCR| + \sum_{j = 1}^r |T_j| \\
&= \frac{\committeeSize}{n} \cdot (n - |N_\GCR|)
= |N_\RD| \cdot \frac{\committeeSize}{n}.
\qedhere
\end{align*}
\end{proof}

We are now ready to establish \Cref{thm:sUFS+FJR}.

\begin{proof}[Proof of \Cref{thm:sUFS+FJR}]
We start by showing that the randomized committee~$(\boldsymbol{\lambda}, \mathbf{W})$ (and the fractional committee~$\vec{p}$ that it implements) returned by \Cref{alg:sUFS+FJR} is feasible and each of the integral committees in the support satisfies ex-post FJR.
First of all, we have $|W_\GCR| \leq \committeeSize$~\citep{PPS21a}.
Next, note that $\sum_{i \in N} b_i = \sum_{i \in N_\GCR} b_i + \sum_{i \in N_\RD} b_i = \sum_{i \in N_\GCR} b_i + k_\RD - \sum_{i\in N_\GCR} b_i = k_\RD$.
By the feasibility reasoning used in the proof of \Cref{thm:GFS+EJR}, each integral committee in the randomized committee returned by \cref{alg:sUFS+FJR:call-bwmes} of \Cref{alg:sUFS+FJR} is of size~$k_\RD$.
Lastly, the randomized committee~$(\boldsymbol{\lambda}, \mathbf{W})$ is composed of outcomes of size $k_\RD + |W_\GCR| = k - |W_\GCR| + |W_\GCR| = k$ since each $W_j$ for $j\in [s]$ is a committee of candidates disjoint from $W_\GCR$.
The fact that $(\boldsymbol{\lambda}, \mathbf{W})$ satisfies ex-post FJR is immediate since each integral committee includes $W_\GCR$, which satisfies FJR \citep{PPS21a}.

For the remainder of the proof, we will show that the fractional committee~$\vec{p}$ satisfies ex-ante Strong UFS.
We let~$y_{ij}$ denote the amount voter $i\in N$ spent on candidate $j \in C \setminus W_\GCR$ in the call to \BWMES (\Cref{alg:GFS+EJR}) in \cref{alg:sUFS+FJR:call-bwmes} of \Cref{alg:sUFS+FJR} and refer to the integral committee obtained during the \MES portion (i.e., \cref{ALG:GFS+EJR:MES}) of \Cref{alg:GFS+EJR} as~$W_\MES$.

We note that each unanimous group in~$N$ either forms a subset of or is disjoint from \emph{any} weakly cohesive group encountered in \cref{alg:sUFS+FJR:gcr} of \Cref{alg:sUFS+FJR}, because for any unanimous group partially included in a weakly $(\beta, T)$-cohesive group, the excluded voters can be added to form a weakly $(\beta, T)$-cohesive group of strictly greater size.
In the following, we show Strong UFS is satisfied for unanimous groups in~$N_\GCR$ and in~$N_\RD$ separately, and begin with voters~$N_\RD$.
Fix any unanimous group~$S \subseteq N_\RD$ and denote their approval set as~$A_S$, i.e., $A_S = A_i$ for all~$i \in S$.
If $A_S \subseteq W_\GCR \cup W_\MES$, then $\sum_{c \in A_S} p_c = |A_S| \geq \min \left\{ |S| \cdot \frac{k}{n}, |A_S| \right\}$, meaning that voters~$S$ are satisfied with respect to Strong UFS.
Otherwise, $A_S \setminus (W_\GCR \cup W_\MES) \neq \emptyset$, and every voter in~$S$ will spend their entire budget on approved candidates during the call to \Cref{alg:GFS+EJR}.
As such, for each~$i \in S$, it holds that
\[
u_i(\vec{p}) = \sum_{c \in A_S} p_c \geq \sum_{c \in A_S} \sum_{i \in N} y_{ic} \geq \sum_{i \in S} \sum_{c \in A_S} y_{ic} = \sum_{i \in S} b_i \geq |S| \cdot \frac{\committeeSize}{n},
\]
where the last transition is due to \Cref{lem:brd-lower-bound}.

We now proceed to show that each unanimous group in $N_\GCR$ is satisfied with respect to Strong UFS.
Fix any unanimous group~$N_j^z$ belonging to the weakly cohesive group encountered in the $j$-th step of \GCR.
First, let the constant~$b$ be such that $b = b_i$ for all~$i \in N_j^z$.
We note that such a constant exists since the definition of~$b_i$ in \cref{alg:sUFS+FJR:b-jz} depends only on~$j$ and~$z$.
If $b = 0$, then by \Cref{claim:b-jz:zero}, we have the following for all~$i \in N_j^z$:
\[
u_i(\vec{p}) \geq |A_i\cap W_\GCR| \geq \beta_j \geq |N_j^z| \cdot \frac{k}{n}.
\]
Lastly, we turn on our attention to the case in which $b > 0$.
Fix~$i \in N_j^z$.
If $A_i \subseteq (W_\GCR \cup W_\MES)$, we have Strong UFS.
Otherwise, voters in $N_j^z$ spend their budgets~$b_i$ entirely on approved candidates.
Thus,
\[
u_i(\vec{p}) \geq \beta_j + \sum_{l \in N_j^z} b_l = \beta_j + |N_j^z| \cdot \left( \frac{k}{n} - \frac{\beta_j}{|N_j^z|} \right) = |N_j^z| \cdot \frac{k}{n}.
\qedhere
\]
\end{proof}

Given the main result of \Cref{sec:GFS+EJR}, it is natural to wonder whether \Cref{alg:sUFS+FJR} satisfies GFS.
We show below that this is not the case.

\begin{example}[\Cref{alg:sUFS+FJR} does not guarantee GFS]
Consider an instance with $n = 3$, $\committeeSize = 2$, and the following approval preferences:
\[
\approvalBallotOfAgent{1} = \{c_1, c_2\} \qquad
\approvalBallotOfAgent{2} = \{c_1, c_3\} \qquad
\approvalBallotOfAgent{3} = \{c_4\}.
\]

Since $\committeeSize = 2$, it must be that $|T| \leq 2$ for any weakly $(\beta, T)$-cohesive group.
It can be verified that for any~$S \subseteq [3]$ and~$T \subseteq \{c_1, c_2, c_3, c_4\}$, the group of voters~$S$ is not weakly $(2, T)$-cohesive.
This holds since such a group would necessarily be of size~$n$ and it is apparent that there is no set of two candidates on which all voters agree.
Note, however, that the group of voters~$\{1, 2\}$ is $1$-cohesive for~$c_1$ (and thus also weakly $(1, \{c_1\})$-cohesive).

The \GCR portion of \Cref{alg:sUFS+FJR} adds voters~$\{1, 2\}$ to~$N_\GCR$, includes candidate~$\{c_1\}$ in~$W_\GCR$, and terminates since no weakly cohesive groups remain.
Since voters~$\{1, 2\}$ have distinct preferences and thus form unanimous groups of size one, from \cref{alg:sUFS+FJR:b-jz}, we have for each~$i \in \{1, 2\}$ that $b_i = \max \{0, 2/3 - 1\} = 0$.
Thus, $b_3 = 1$ due to \cref{alg:sUFS+FJR:b-rd}, which will be spent on~$c_4$, the only candidate she approves.
The resulting fractional committee is $\vec{p} = (1, 0, 0, 1)$, which does not satisfy GFS with respect to the group~$S = \{1, 2\}$ since:
\[
1 = \sum_{c \in \bigcup_{i \in S} \approvalBallotOfAgent{i}} p_c < \sum_{i \in S} \frac{1}{n} \cdot \min \{\committeeSize, |\approvalBallotOfAgent{i}|\} = \frac{4}{3}.
\]
\end{example}

Since \Cref{alg:sUFS+FJR} does not satisfy GFS, we may look to the analysis of \Cref{thm:GFS+EJR} for inspiration.
There, we used the ingredient of voter-specific candidate payments~$y_{ij}$, which are critically not part of the output of \GCR.
It then seems intriguing to consider whether we can retrospectively determine such payments for~$W_\GCR$.
These payments bear a close resemblance to what is called a \emph{price system}, and are highly related to an axiom called \emph{priceability}~\citep{PeSk20a}.
Indeed, it is known that a price system can always be constructed for the output from \GCR \citep{PPS21a}.
While this seems promising, the mere existence of a price system does not immediately provide a non-zero lower bound on the total amount each voter contributes to approved candidates, as shown in \Cref{claim:LB-spending} and used to prove \Cref{thm:GFS+EJR}.
Thus, it does not immediately follow from the same argument given for \Cref{thm:GFS+EJR} that \GCR can be completed to a GFS committee.
It remains an intriguing open question whether there exists a randomized committee which satisfies ex-ante GFS and ex-post FJR.

\section{Conclusion}

In this work, we have initiated the best-of-both-worlds paradigm in the context of committee voting, which allows us to achieve both ex-ante and ex-post fairness.
We first generalized \emph{fair share} axioms from the single-winner probabilistic voting literature, and subsequently gave algorithms which satisfied these properties in conjunction with existing ex-post fairness properties. As mentioned, an immediate open question left by our work is the compatibility of ex-ante GFS and ex-post FJR.
More broadly, our work motivates the question of whether we can achieve other ex-ante properties such as Pareto efficiency or fractional core~\citep{FGM16,CJMW20} while ensuring standard ex-post properties such as JR, PJR and EJR.

Approval-based committee voting is but one of many social choice settings of interest.
Others include multiple referenda~\citep{BKZ97a}, public decision making~\citep{CFS17a}, and participatory budgeting~\citep{AzSh20a,ReMa23}.
What new challenges does implementation present in more complex settings such as participatory budgeting, which involves candidate costs and budget constraints?
We hope that our work serves as an invitation for further research applying the best-of-both-worlds perspective to social choice problems.

\section*{Acknowledgements}

This work was partially supported by the NSF-CSIRO grant on ``Fair Sequential Collective Decision-Making'' and the ARC Laureate Project FL200100204 on ``Trustworthy AI''.
We would like to thank the anonymous reviewers for their valuable feedback.

\bibliographystyle{plainnat}
\bibliography{bibliography}

\clearpage
\appendix

\section{Relations Between Ex-ante Fairness Notions (Cont.)}
\label{app:relations-ex-ante-fairness}

We continue our discussions regarding to the relations between the proposed ex-ante fairness notions.
First of all, recall that Random Dictator satisfies GFS but not Strong IFS.
Therefore, GFS (and thus UFS) does not imply Strong IFS, let alone Strong UFS.

Next, Strong IFS does not imply UFS, which can be seen from the following example.

\begin{example}[Strong IFS does not imply UFS]
Consider an instance with $n = 3$, $k = 2$, and approval preferences
\[
A_1 = A_2 = \{c_1, c_2\} \qquad A_3 = \{c_3, c_4\}.
\]
Observe that the fractional committee~$\vec{p} = \left( \frac{2}{3}, 0, 1, \frac{1}{3} \right)$ satisfies Strong IFS.
However, $\vec{p}$ does not satisfy UFS with respect to the group~$S = \{1, 2\}$ since
\[
\frac{2}{3} = \sum_{c \in \{c_1, c_2\}} p_c < \frac{|S|}{n} \cdot \min \{k, |A_i|\} = \frac{4}{3}.
\]
\end{example}

Finally, we show Strong UFS (and thus Strong IFS) does not imply GFS.
It suffices to show that Strong UFS does not imply GFS.

\begin{example}[Strong UFS does not imply GFS]
Consider an instance with $n = 3$, $k=2$, and the following approval preferences:
\[
A_1 = \{c_1, c_2\} \qquad A_2 = \{c_1, c_3\} \qquad A_3 = \{c_4, c_5\}.
\]
Observe that the fractional committee~$\vec{p} = \left( \frac{2}{3}, 0, 0, 1, \frac{1}{3} \right)$ satisfies Strong UFS.
However, $\vec{p}$ does not satisfy GFS with respect to the group $S = \{1, 2\}$ since
\[
\frac{2}{3} = \sum_{c \in \bigcup_{i \in S} \approvalBallotOfAgent{i}} p_c < \sum_{i \in S} \frac{1}{n} \cdot \min \{\committeeSize, |\approvalBallotOfAgent{i}|\} = \frac{4}{3}.
\]
\end{example}
\end{document}